\RequirePackage{amsmath}
\RequirePackage{amssymb}

\documentclass{llncs}

\usepackage[utf8]{inputenc}
\usepackage[T1]{fontenc}
\usepackage[english]{babel}
\usepackage{lmodern}
\usepackage{microtype}

\usepackage[algo2e,ruled,vlined]{algorithm2e}
\usepackage{float}
\usepackage{wrapfig}
\restylefloat{figure}

\usepackage{graphicx}
\usepackage{hyperref}
\usepackage{multirow}
\usepackage{subfig}
\usepackage{xcolor}

\usepackage[inline]{enumitem}
\usepackage[disable]{todonotes}

\usepackage[capitalize, nameinlink]{cleveref}
\crefname{section}{Sec.}{Sec.}
\crefname{algorithm}{Alg.}{Alg.}
\crefname{figure}{Fig.}{Fig.}
\crefname{proposition}{Prop.}{Prop.}
\crefname{table}{Table}{Tables}
\crefname{definition}{Def.}{Def.}
\crefname{theorem}{Thm.}{Thm.}

\renewcommand{\paragraph}[1]{\medskip \noindent {\bf #1.}}

\usepackage{listings}
\usepackage{tikz,pgfplots}

\usepackage{pifont}

\newcommand{\abv}{\textsc{ABV}}

\newcommand{\qf}{\textsc{QF}}

\newcommand{\sat}{\textsc{sat}}
\newcommand{\timeout}{\textsc{timeout}}
\newcommand{\unsat}{\textsc{unsat}}
\newcommand{\unknown}{\textsc{unknown}}

\newcommand{\sic}{\textsc{sic}}
\newcommand{\wic}{\textsc{wic}}

\newcommand{\A}{\mathcal{A}}
\newcommand{\B}{\mathcal{B}}
\newcommand{\I}{\mathcal{I}}
\newcommand{\J}{\mathcal{J}}

\newcommand{\Q}{\mathcal{Q}}
\newcommand{\R}{\mathcal{R}}
\newcommand{\T}{\mathcal{T}}

\newcommand{\X}{\mathcal{X}}

\renewcommand{\S}{\mathcal{S}}

\newcommand{\binsec}{\textsc{Binsec}}
\newcommand{\coq}{Coq}

\newcommand{\tfml}{\textsc{Tfml}}

\newcommand{\overall}{solveQ}
\newcommand{\generic}{inferSIC}
\newcommand{\theory}{theorySIC}
\newcommand{\solve}{solveQF}

\newcommand{\code}[1]{\texttt{#1}}
\newcommand{\taint}[1]{#1^\bullet}
\newcommand{\set}[1]{\boldsymbol{#1}}
\newcommand{\val}[1]{\mathtt{#1}}

\newcommand{\urlcoq}{https://frama.link/DnMWPM7q}

\usepackage{stmaryrd}
\newcommand{\sem}[1]{\llbracket #1 \rrbracket}

\newcommand{\smtar}[2]{\mbox{Array}\,{#1}\,{#2}}
\newcommand{\smtbv}[1]{\mbox{BitVec}\,{#1}}
\newcommand{\smtbl}{\mbox{Bool}}

\newcommand{\ite}{\mbox{ite}}

\newcommand{\concat}{\mbox{concat}}
\newcommand{\extract}{\mbox{extract}}

\newcommand{\blnot}{\neg}
\newcommand{\blimply}{\Rightarrow}
\newcommand{\bland}{\wedge}
\newcommand{\blor}{\vee}
\newcommand{\blxor}{\otimes}

\newcommand{\bvnot}{\neg}
\newcommand{\bvor}{\vee}
\newcommand{\bvand}{\wedge}
\newcommand{\bvadd}{+}
\newcommand{\bvsub}{-}
\newcommand{\bvmul}{\times}
\newcommand{\bvdiv}{\div}
\newcommand{\bvshl}{\ll}
\newcommand{\bvshr}{\gg}

\newcommand{\bvlt}{<}
\newcommand{\bvgt}{>}
\newcommand{\bvle}{\leq}
\newcommand{\bvge}{\geq}

\newcommand{\select}{\mbox{select}}
\newcommand{\store}{\mbox{store}}
\newcommand{\smtlib}{\textsc{SMT-LIB}}

\title{Model Generation for Quantified Formulas: \\ A Taint-Based Approach}

\author{
  Benjamin~Farinier\inst{1,2} \and
  Sébastien~Bardin\inst{1} \and\\
  Richard~Bonichon\inst{1} \and
  Marie-Laure~Potet\inst{2}
}

\institute{
  CEA, LIST, Software Safety and Security Lab, Université Paris-Saclay, France
  \\\email{firstname.lastname@cea.fr}
  \and Univ.\ Grenoble Alpes, Verimag, France
  \\\email{firstname.lastname@univ-grenoble-alpes.fr}
}

\begin{document}

\maketitle

\begin{abstract}\vspace{-8pt}
  We focus in this paper on generating models of quantified first-order formulas
  over built-in theories, which is paramount in software verification and bug
  finding.
  While standard methods are either geared toward proving the absence of
  solution
  %
  %
  or targeted to specific theories,
  %
  %
  %
  we propose a generic approach  based on a reduction to the quantifier-free
  case.
  %
  %
  %
  Our technique
  %
  %
  allows thus to reuse all the efficient machinery developed for that context.
  Experiments show a substantial improvement over  state-of-the-art methods.
  %
  %
  %
  %
\end{abstract}

\section{Introduction}

\paragraph{Context}
Software verification methods have come to rely increasingly on reasoning over
logical formulas modulo theory.
%
In particular, the ability to generate  models  (i.e., find solutions)   of a
formula is of utmost importance, typically in the context of bug finding or
intensive testing --- symbolic execution \cite{GLM12} or bounded model checking
\cite{B09}.
%
%
%
Since {\it quantifier-free first-order formulas} on well-suited theories are
sufficient to represent many reachability properties of interest, the
Satisfiability Modulo Theory (SMT) \cite{BSST09,KS08} community has primarily
dedicated itself to designing solvers able to efficiently handle such problems.

Yet, universal quantifiers are sometimes needed, typically when considering
preconditions or code abstraction.
%
%
Unfortunately, most theories handled by SMT-solvers are undecidable in the
presence of universal quantifiers.
There exist dedicated methods for a few decidable quantified theories, such as
Presburger arithmetic \cite{BKRW11} or the array property fragment \cite{BMS06},
but there is no general and effective enough approach for the model generation
problem  over universally quantified formulas.
%
%
Indeed, generic solutions for quantified formulas involving heuristic
instantiation and refutation are geared at proving the unsatisfiability of
a formula (i.e., absence of solution) \cite{MB07,GM09}, while recent proposals
such as local theory extensions \cite{BRKBW15}, finite instantiation
\cite{RTGK13,RTGKDB13} or model-based instantiation \cite{RDKTB15,GM09} either
are too narrow in scope, or handle quantifiers on free sorts only, or restrict
themselves to finite models, or may get stuck in infinite refinement loops.

\paragraph{Goal and challenge} Our goal is to propose a generic and efficient
approach to the model generation problem over
arbitrary quantified formulas  with support for theories commonly found in
software verification.
Due to the huge effort made by the community to produce state-of-the-art solvers
for quantifier-free theories (\emph{\qf-solvers}), it is highly desirable for
this solution to be compatible with current leading decision procedures, namely
SMT approaches.

\paragraph{Proposal}
%
%
Our approach turns a quantified formula into a quantifier-free formula with the
guarantee that any model of the later contains a model of the former.
The benefits are threefold: the transformed formula is  easier to solve,
it can be sent to  standard \qf-solvers,  and  a model for the initial formula
is deducible from a model of the transformed one.
The idea is to ignore quantifiers but strengthen the quantifier-free part of the
formula with an  \emph{independence condition} constraining models to be
independent from the (initially) quantified variables.
%
%
%
%

\paragraph{Contributions} This paper makes the following contributions:
\begin{description}[topsep=0pt, leftmargin=.4cm]
  \item We propose a novel and generic framework for model generation of
    quantified formula (\cref{sec:generic}, \cref{algorithm:overall})  relying
    on the inference of \emph{sufficient independence condition}
    (\Cref{section:preliminary}).
    We prove its \emph{correctness} (\Cref{prop:algo_correct}, mechanized in
    \coq) and its \emph{efficiency} under reasonable assumptions
    (\Cref{prop:size_bound,prop:complexity_bound}).
    Especially our approach implies only a linear overhead in the formula size.
    We also briefly study its \emph{completeness}, related to the notion of {\it
    weakest independence condition}.
  \item We define a taint-based procedure for the inference of independence
    conditions (\Cref{section:taint}), composed of a theory-independent core
    (\Cref{algorithm:taint:generic}) together with theory-dependent refinements.
    We propose such refinements for a large class of operators
    (\Cref{sec:taint:rabsorbant}), encompassing notably arrays and bitvectors.
  \item Finally, we present a concrete implementation of our method specialized
    on arrays and bitvectors (\Cref{section:benchmarks}).
    %
    \todo{plus assertif : c important, on fait des vrais xps de cas, Z3 \&co ne
      savent pas faire, on complemente, boolector meilleur, etc. quantified
      reachability}
    Experiments on \smtlib\ benchmarks and software verification problems
    demonstrate that we are able not only to very effectively lift
    quantifier-free decision procedures to the quantified case,
    %
    %
    but also to supplement recent advances, such as finite or model-based
    quantifier instantiation \cite{RTGK13,RTGKDB13,RDKTB15,GM09}.
    Indeed, we concretely supply SMT solvers with the ability to efficiently
    address an extended set of software verification questions.
\end{description}

\paragraph{Discussions}
%
%
\todo{a mettre en avant : 1-completely new way of dealing with quantif
(complementaire de other ways, open perspective), 2- apply it on a hard problem
(quantif arrays et uf) sur des cas realistes (quantif reach) : ca assure.  3-
pour le moment : bien sur ``containers'' arry, ite, uf. Next : ``data''}
Our approach supplements state-of-the-art model generation on quantified
formulas by providing a more generic handling of satisfiable problems.
We can deal with quantifiers on any sort and we are not restricted to finite
models.
Moreover, this is a lightweight preprocessing approach requiring a single call
to the underlying quantifier-free solver.
%
%
%
%
The method also extends to \emph{partial} elimination of universal quantifiers,
or reduction to \emph{quantified-but-decidable} formulas
(\Cref{section:extensions}).
While techniques \emph{a la} E-matching  allow to lift quantifier-free solvers
to the unsatisfiability checking of quantified formulas, this works provides a
mechanism to lift them to the satisfiability checking and model generation of
quantified formulas, yielding a more symmetric handling of quantified formulas
in SMT.
This new approach paves the way to future developments such as the
definition of more precise inference mechanisms of independence conditions, the
identification of subclasses for which inferring weakest
independence conditions is feasible, and the combination with other quantifier
instantiation techniques.

%
%
%

\section{Motivation} \label{sec:motivation}

Let us take the code sample in \Cref{fig:motivation} and suppose we want to
reach  function \code{analyze\_me}.
%
%
For this purpose, we need a model (a.k.a., solution) of the reachability
condition $\phi \triangleq ax+b>0$, where $a$, $b$ and $x$ are symbolic
variables associated to the program variables \code{a}, \code{b} and \code{x}.
However, while the values of \code{a} and \code{b} are user-controlled, the
value of \code{x} is not.
Therefore if we want to reach \code{analyze\_me} in a reproducible manner,
we actually need a model of $\phi_{\forall} \triangleq \forall{x}.ax+b>0$, which
\emph{involves universal quantification}.
While this specific formula is simple,
%
%
%
model generation for quantified formulas is notoriously difficult:
PSPACE-complete for booleans, undecidable for uninterpreted functions or arrays.

\vspace{-\intextsep}
\begin{figure}
  \begin{minipage}{0.4\textwidth}
    \begin{lstlisting}
int main () {
  int a = input ();
  int b = input ();

  int x = rand ();

  if (a * x + b > 0) {
    analyze_me();
  }
  else {
    ...;
  }
}
    \end{lstlisting}
  \end{minipage}
  \begin{minipage}{0.6\textwidth}
    \[
      \begin{array}{rl}
        \multicolumn{2}{l}{\mbox{\bf Quantified reachability condition}}\smallskip\\
        (1) & \forall{x}.{ax+b>0}\medskip\\
        \multicolumn{2}{l}{\mbox{{\bf Taint variable constraint}}}\\
        (2) &
        \taint{a}
        \wedge\taint{b}
        \wedge\neg\left(\taint{x}\right)
        \hfill(\taint{a},\taint{b},\taint{x}:\mbox{fresh boolean})
        \medskip\\
        \multicolumn{2}{l}{\mbox{\bf Independence condition}}\smallskip\\
        (3) &
        \left(\left(\taint{a}\wedge\taint{x}\right)
        \vee\left(\taint{a}\wedge{a=0}\right)
        \vee\left(\taint{x}\wedge{x=0}\right)\right)\wedge\taint{b}
        \smallskip\\
        (4) &
        \left(\left(\top\wedge\bot\right)
        \vee\left(\top\wedge{a=0}\right)
        \vee\left(\bot\wedge{x=0}\right)\right)\wedge\top
        \smallskip\\
        (5) & {a=0}\medskip\\
        \multicolumn{2}{l}{\mbox{\bf Quantifier-free approximation of (1)}}\smallskip\\
        (6) & \left(ax+b>0\right)\wedge\left(a=0\right)\\
      \end{array}
    \]
  \end{minipage}
  \caption{Motivating example}
  \label{fig:motivation}
\end{figure}
\vspace{-\intextsep}

\paragraph{Reduction to the quantifier-free case through independence}
%
%
We propose to ignore the universal quantification over $x$, but \emph{restrict
models to those which do not depend on $x$}.
For example, model $\{\val{a}=1, \val{x}=1, \val{b}=0\}$ does depend on $x$, as
taking $x=0$ invalidates the formula, while model $\left\lbrace\val{a}=0,
\val{x}=1, \val{b}=1\right\rbrace$ is \emph{independent of} $x$.
We call constraint $\psi \triangleq (a=0)$ an \emph{independence condition}: any
interpretation of $\phi$ satisfying $\psi$ will be independent of $x$, and
therefore a model of $\phi \wedge \psi$ will give us a model of
$\phi_{\forall}$.
%

\paragraph{Inference of independence conditions through tainting}
\Cref{fig:motivation} details in its right part a way to infer such independence
conditions.
%
%
Given a quantified reachability condition~(1), we first associate to every
variable $v$ a (boolean) \emph{taint variable} $\taint{v}$ indicating whether
the solution may depend on $v$ (value $\top$) or not (value $\bot$).
%
%
%
Here, $\taint{x}$ is set to $\bot$,  $\taint{a}$ and $\taint{b}$ are set to
$\top$ (2).
%
%
An independence condition~(3) --- a formula modulo theory ---   is then
constructed using both initial and taint variables.
We extend  taint constraints to  terms,  $\taint{t}$ indicating here whether $t$
may depend on $x$ or not, and we  require  the top-level term (i.e., the
formula) to be tainted to $\top$ (i.e., to be indep.~from $x$).
Condition (3) reads as follows: in order to enforce that $\taint{(ax+b>0)}$
holds, we enforce that $\taint{(ax)}$ and $\taint{b}$ hold, and for
$\taint{(ax)}$ we require that either $\taint{a}$ and $\taint{x}$ hold, or
$\taint{a}$ holds and $a=0$ (absorbing the value of $x$), or the symmetric case.
We see that $\taint{\cdot}$ is defined recursively and combines a {\it
systematic part} (if $\taint{t}$ holds then $\taint{f(t)}$ holds, for any $f$)
with a {\it theory-dependent part} (here, based on $\times$).
%
%
%
%
After simplifications~(4), we obtain $a=0$ as an independence condition~(5)
which  is adjoined to the reachability condition freed of its universal
quantification~(6).
A \qf-solver provides a model of (6) (e.g., $\left\lbrace
\val{a}=0,\val{b}=1,\val{x}=5 \right\rbrace$), lifted into a model of (1) by
discarding the valuation of $x$ (e.g., $\left\lbrace \val{a}=0,\val{b}=1
\right\rbrace$).

In this specific example the inferred independence condition (5) is the most
generic one and (1) and (6) are equisatisfiable.
Yet, in general  it may be an under-approximation, constraining the variables
more than needed and yielding a correct but incomplete decision method: a model
of (6) can still be turned into a model of (1), but (6) might not have a model
while (1) has.

\section{Notations}

We consider the framework of  many-sorted first-order logic with equality, and
we assume  standard definitions of sorts, signatures and terms.
%
%
%
%
%
%
%
%
Given a tuple of variables $\set{x} \triangleq \left(x_1,\dots,x_n\right)$ and a
quantifier $\Q$ ($\forall$ or $\exists$), we shorten $\Q{x_1}\dots\Q{x_n}.\Phi$
as $\Q\set{x}.\Phi$.
A formula is in \emph{prenex normal form}
%
%
if it is written as $\Q_1{\set{x}_1}\dots\Q_n{\set{x}_n}.\Phi$ with $\Phi$  a
quantifier-free formula.
%
%
A formula is in \emph{Skolem normal form} if it is in prenex normal form with
only universal quantifiers.
We write $\Phi\left(\set{x}\right)$ to denote that the free variables of  $\Phi$
are in $\set{x}$.
Let $\set{t}\triangleq\left(t_1,\dots,t_n\right)$ be a term tuple, we write
$\Phi\left(\set{t}\right)$ for the formula obtained from $\Phi$ by replacing
each occurrence of $x_i$ in $\Phi$ by $t_i$.
%
%
An {\it interpretation} $\I$ associates a domain to each sort of a signature and
a value to each symbol of a formula, and $\sem{\Delta}_{\I}$ denotes the
evaluation of term $\Delta$ over $\I$.
A {\it satisfiability relation} $\models$  between interpretations and formulas
is defined inductively as usual.
A {\it model} of $\Phi$ is an interpretation $\I$ satisfying $\I \models \Phi$.
We sometimes refer to models as ``solutions''.
Formula $\Psi$ \emph{entails}  formula $\Phi$, written $\Psi\models\Phi$, if
every interpretation satisfying $\Psi$ satisfies $\Phi$ as well.
Two formula are equivalent, denoted $\Psi \equiv \Phi$, if they have the same
models.
%
%
%
%
A {\it theory}  $\T\triangleq\left(\Sigma,\set\I\right)$ restricts symbols in
$\Sigma$ to be interpreted in $\set\I$.
The quantifier-free fragment of $\T$ is denoted \qf-$\T$.

%
%

\paragraph{Convention}
Letters $a,b,c\dots$ denote uninterpreted symbols and variables.
Letters $x,y,z\dots$ denote quantified variables.
$\set{a},\set{b},\set{c}$ denote sets of uninterpreted symbols.
$\set{x},\set{y},\set{z}\dots$ denote sets of quantified variables.
Finally, $\val{a},\val{b},\val{c}\dots$ denote valuations of associated (sets
of) symbols.

%
%
\emph{
In the rest of this paper, we assume w.l.o.g.~that all formulas are in Skolem
normal form.
Recall that any formula $\phi$ in classical logic can be normalized into a
formula $\psi$ in Skolem normal form  such that any model of $\phi$ can be
lifted into a model of $\psi$, and vice versa.
This strong relation, much closer to formula equivalence than to formula
equisatisfiability, ensures that our correctness and completeness results all
along the paper hold for arbitrarily quantified formula.
}

\paragraph{\emph{Appendix}}
\emph{
Additional technical details (proofs, experiments, etc.) are provided in
appendix.
This  content will be made  available in an online technical report.
}

\section{Musing with independence}\label{section:preliminary}

\subsection{Independent interpretations, terms and formulas}

%
%
%
%
A solution $(\val{x},\val{a})$ of $\Phi$ does not depend on $\set{x}$ if
$\Phi(\set{x},\set{a})$ is always true or always false, for all possible
valuations of $\set{x}$ as long as $\set{a}$ is set to $\val{a}$.
More formally, we define the independence of an interpretation of $\Phi$
w.r.t.~$\set{x}$ as follows:


\begin{definition}[Independent interpretation]\label{def:indep_interp}
  \begin{itemize}[topsep=0pt]

    \item Let $\Phi\left(\set{x},\set{a}\right)$ a formula with free variables
      $\set{x}$ and $\set{a}$.
      Then an interpretation $\I$ of $\Phi\left(\set{x},\set{a}\right)$ is
      independent of $\set{x}$ if for all interpretations $\J$ equal to $\I$
      except on $\set{x}$, $\I\models\Phi$ if and only if $\J\models\Phi$.

    \item Let $\Delta\left(\set{x},\set{a}\right)$ a term with free variables
      $\set{x}$ and $\set{a}$.
      Then an interpretation $\I$ of $\Delta\left(\set{x},\set{a}\right)$ is
      independent of $\set{x}$ if for all interpretations $\J$ equal to $\I$
      except on $\set{x}$, $\sem{\Delta\left(\set{x},\set{a}\right)}_{\I} =
      \sem{\Delta\left(\set{x},\set{a}\right)}_{\J}$.

  \end{itemize}
\end{definition}

Regarding formula ${ax+b>0}$ from \Cref{fig:motivation}, $\left\lbrace {a=0},
{b=1}, {x=1} \right\rbrace$ is independent of $x$ while $\left\lbrace {a=1},
{b=0}, {x=1} \right\rbrace$ is not.
Considering  term $\left( t\left[{a}\leftarrow{b} \right] \right)[c]$, with $t$
an array written at index $a$ then read at index $c$, $\left\lbrace{a=0},
{b=42}, {c=0}, {t=\left[\dots\right]} \right\rbrace$ is independent of $t$
(evaluates to $42$) while $\left\lbrace {a=0}, {b=1}, {c=2},
{t=\left[\dots\right]} \right\rbrace$ is not (evaluates to $t\left[2\right]$).
%
%
%
%
%
%
%
%
%
We now define independence for formulas and terms.

%

\begin{definition}[Independent formula and term]\label{def:indep_form_term}
  \begin{itemize}[topsep=0pt]

    \item Let $\Phi\left(\set{x},\set{a}\right)$ a formula with free variables
      $\set{x}$ and $\set{a}$.\!
      Then $\Phi\left(\set{x},\set{a}\right)$ is independent of $\set{x}$ if
      $\forall\set{x}.\forall\set{y}. \left(\Phi\left(\set{x},\set{a}\right)
      \Leftrightarrow \Phi\left(\set{y},\set{a}\right)\right)$ holds true for
      any value of $\set{a}$.

    \item Let $\Delta\left(\set{x},\set{a}\right)$ a term with free variables
      $\set{x}$ and $\set{a}$.
      Then $\Delta\left(\set{x},\set{a}\right)$ is independent of $\set{x}$ if
      $\forall\set{x}.\forall\set{y}. \left(\Delta\left(\set{x},\set{a}\right) =
      \Delta\left(\set{y},\set{a}\right)\right)$ holds true for any value of
      $\set{a}$.

  \end{itemize}
\end{definition}

\Cref{def:indep_form_term} of formula and term independence is far stronger than
\Cref{def:indep_interp} of interpretation independence.
Indeed, it can easily be checked that if a formula $\Phi$ (resp.\ a term
$\Delta$) is independent of  $\set{x}$, then any interpretation of $\Phi$
(resp.\ $\Delta$) is independent of  $\set{x}$.
However, the converse is false as formula $ax+b>0$ is not independent of $x$,
but has an interpretation $\left\lbrace{a=0},{b=1},{x=1}\right\rbrace$ which is.


\subsection{Independence conditions} \label{sec:SIC}

Since it is rarely the case that a formula (resp.\ term) is independent from a
set of variables $\set{x}$, we are interested in \emph{Sufficient Independence
Conditions}.
These conditions are additional constraints that can be added to a formula
(resp.\ term) in such a way that they make the formula (resp.\ term) independent
of $\set{x}$.

\begin{definition}[Sufficient Independence Condition  (SIC)]\label{sic}\hfill
  \begin{itemize}[topsep=0pt]

    \item A Sufficient Independence Condition for a formula
      $\Phi\left(\set{x},\set{a}\right)$ with regard to $\set{x}$ is a formula
      $\Psi\left(\set{a}\right)$ such that $\Psi\left(\set{a}\right) \models
      ({\forall\set{x}.\forall\set{y}.  \Phi\left(\set{x},\set{a}\right)
      \Leftrightarrow \Phi\left(\set{y},\set{a}\right)})$.

    \item A Sufficient Independence Condition for a term
      $\Delta\left(\set{x},\set{a}\right)$ with regard to $\set{x}$, is a
      formula $\Psi\left(\set{a}\right)$ such that $\Psi\left(\set{a}\right)
      \models ({\forall\set{x}.\forall\set{y}.
      \Delta\left(\set{x},\set{a}\right) =
      \Delta\left(\set{y},\set{a}\right)})$.

  \end{itemize}
\end{definition}

We denote by $\sic_{\Phi,\set{x}}$ (resp.\ $\sic_{\Delta,\set{x}}$) a Sufficient
Independence Condition for a formula $\Phi\left(\set{x},\set{a}\right)$ (resp.\
for a term $\Delta\left(\set{x},\set{a}\right)$) with regard to $\set{x}$.
For example,  $a=0$ is a $\sic_{\Phi,x}$ for  formula $\Phi \triangleq
{ax+b>0}$, and $a=c$ is a $\sic_{\Delta,t}$ for term $\Delta \triangleq
\left(t\left[{a}\leftarrow{b}\right]\right)[c]$.
Note that $\bot$ is always a \sic, and that \sic\ are closed under $\wedge$ and
$\vee$.
%
%
%
%
%
\Cref{prop:sic_correct} clarifies the interest of \sic\ for model generation.
%

\begin{proposition}[Model generalization]\label{prop:sic_correct}
  Let $\Phi\left(\set{x},\set{a}\right)$  a formula and $\Psi$  a
  $\sic_{\Phi,\set{x}}$.
  If there exists an interpretation $\left\lbrace\val{x},\val{a}\right\rbrace$
  such that $\left\lbrace\val{x},\val{a}\right\rbrace \models
  \Psi\left(\set{a}\right) \wedge \Phi\left(\set{x},\set{a}\right)$, then
  $\left\lbrace\val{a}\right\rbrace \models
  \forall\set{x}.\Phi\left(\set{x},\set{a}\right)$.
\end{proposition}

\begin{proof}[sketch of]
  \Cref{proof:sic_correct}.
\end{proof}


For the sake of completeness, we introduce now the notion of \emph{Weakest
Independence Condition} for a formula $\Phi\left(\set{x},\set{a}\right)$ with
regard to $\set{x}$ (resp.\ a term $\Delta\left(\set{x},\set{a}\right)$).
We will denote such conditions $\wic_{\Phi,\set{x}}$ (resp.\
$\wic_{\Delta,\set{x}}$).

\begin{definition}[Weakest Independence Condition (WIC)]\label{wic} \hfill
  \begin{itemize}[topsep=0pt]

    \item A Weakest Independence Condition for a formula
      $\Phi\left(\set{x},\set{a}\right)$ with regard to $\set{x}$  is a
      $\sic_{\Phi,\set{x}}$ $\Pi$ such that, for any other $\sic_{\Phi,\set{x}}$
      $\Psi$, $\Psi\models\Pi$.

    \item A Weakest Independence Condition for a term
      $\Delta\left(\set{x},\set{a}\right)$ with regard to $\set{x}$  is a
      $\sic_{\Delta,\set{x}}$ $\Pi$ such that, for any other
      $\sic_{\Delta,\set{x}}$ $\Psi$, $\Psi\models\Pi$.

  \end{itemize}
\end{definition}

%
Note that $\Omega \triangleq \forall\set{x}.\forall\set{y}.
\left(\Phi\left(\set{x},\set{a}\right) \Leftrightarrow
\Phi\left(\set{y},\set{a}\right)\right)$ is always a $\wic_{\Phi,\set{x}}$, and
any formula $\Pi$ is a $\wic_{\Phi,\set{x}}$ if and only if $\Pi \equiv \Omega$.
%
%
Therefore all syntactically different \wic\ have the same semantics.
%
%
%
As an example, both \sic\ $a=0$ and $a=c$ presented earlier are $\wic$.
%
%
\Cref{prop:wic_complete} emphasizes the interest of \wic\ for model generation.

\begin{proposition}[Model specialization]\label{prop:wic_complete}
  Let $\Phi\left(\set{x},\set{a}\right)$  a formula and $\Pi(\set{a})$  a
  $\wic_{\Phi,\set{x}}$.
  If there exists an interpretation $\left\lbrace \val{a}\right \rbrace$ such
  that $\left\lbrace \val{a} \right\rbrace \models \forall\set{x}.
  \Phi\left(\set{x},\set{a}\right)$, then $\left\lbrace \val{x},\val{a}
  \right\rbrace \models \Pi\left(\set{a}\right) \wedge
  \Phi\left(\set{x},\set{a}\right)$ for any valuation $\val{x}$  of $\set{x}$.
\end{proposition}

\begin{proof}[sketch of]
  \Cref{proof:wic_complete}.
\end{proof}

From now on, our goal is to infer from a formula $\forall\set{x}.
\Phi\left(\set{x},\set{a}\right)$ a $\sic_{\Phi,\set{x}}$
$\Psi\left(\set{a}\right)$, find a model for $\Psi\left(\set{a}\right) \wedge
\Phi\left(\set{x},\set{a}\right)$ and generalize it.
This $\sic_{\Phi,\set{x}}$ should be as weak ---~in the sense ``less
coercive''~--- as possible, as otherwise $\bot$ could always be used, which
would not be very interesting for our overall purpose.

For the sake of simplicity, previous definitions omit to mention the theory  to
which the \sic\ belongs.
If the theory $\T$ of the quantified formula is decidable we can always choose
$\forall\set{x}.\forall\set{y}. \left(\Phi\left(\set{x},\set{a}\right)
\Leftrightarrow \Phi\left(\set{y},\set{a}\right)\right)$ as a \sic,
%
%
but it is simpler to directly use a $\T$-solver.
\emph{The challenge is, for formulas in an undecidable theory $\T$, to find a
non-trivial \sic\ in its quantifier-free fragment \qf-$\T$.}
%

Under this constraint, we cannot expect a systematic construction of $\wic$, as
it would allow to decide the satisfiability of any quantified theory with a
decidable quantifier-free fragment.
%
%
Yet informally, the closer a \sic\ is to be a \wic, the closer our approach is
to completeness.
Therefore this notion might be seen as a fair gauge of the quality of a $\sic$.
\emph{Having said that, we leave a deeper study on the inference of \wic\ as
future work.}

\section{Generic framework for SIC-based model generation}
\label{section:generic}
\label{sec:generic}

%
We describe now our overall approach.
\Cref{algorithm:overall} presents our \sic-based generic framework for model
generation (\Cref{section:solver}).
Then, \Cref{algorithm:taint:generic} proposes a taint-based approach for \sic\
inference (\Cref{section:taint}).
%
%
Finally, we discuss complexity and efficiency issues (\Cref{section:complexity})
%
%
and detail  extensions (\Cref{section:extensions}), such as partial elimination.

\emph{From now on, we do not distinguish anymore between terms and formulas,
their treatment being symmetric, and we call targeted variables the variables we
want to be independent of.}

\subsection{SIC-based model generation}\label{section:solver}

\vspace{-\intextsep}
\begin{algorithm2e}
  \DontPrintSemicolon
  \SetKwFunction{FOverall}{\overall}
  \SetKwFunction{FGeneric}{\generic}
  \SetKwFunction{FSolve}{\solve}
  \SetKwProg{Fn}{Function}{:}{}
  \SetKwProg{KwPara}{Parameter:}{}{}
  \SetKwSwitch{Match}{With}{}{match}{}{with}{}{}{}

  \KwPara{\FSolve}{
    \KwIn{$\Phi(v)$ a formula in  \qf-$\T$}
    \KwOut{$\sat\left(\val{v}\right)$ with $\val{v}\models\Phi$, \unsat\ or \unknown}
  }\smallskip

  \KwPara{\FGeneric}{
    \KwIn{$\Phi$ a formula in  \qf-$\T$, and $\set{x}$ a set of targeted variables}
    \KwOut{$\Psi$ a  $\sic_{\Phi,\set{x}}$ in \qf-$\T$}
  }\smallskip

  \Fn{\FOverall}{
    \KwIn{$\forall\set{x}.\Phi\left(\set{x},\set{a}\right)$ a universally
    quantified formula over theory $\T$}
    \KwOut{$\sat\left(\val{a}\right)$ with $\val{a} \models \forall\set{x}.
    \Phi\left(\set{x},\set{a}\right)$, \unsat\ or \unknown}
    \smallskip
    {Let $\Psi\left(\set{a}\right) \triangleq
    \FGeneric\left(\Phi\left(\set{x},\set{a}\right),\set{x}\right)$}\\
    \Match{\FSolve$\left(\Phi\left(\set{x},\set{a}\right)\wedge\Psi\left(\set{a}\right)\right)$}{
      \lWith{$\sat\left(\val{x},\val{a}\right)$}{
        \KwRet{$\sat\left(\val{a}\right)$}}
      \With{\unsat}{
        \lIf{$\Psi$ is a $\wic_{\Phi,\set{x}}$}{\KwRet{\unsat}}
        \lElse{\KwRet{\unknown}}
      }
      \lWith{\unknown}{\KwRet{\unknown}}
    }}
\caption{SIC-based model generation for quantified formulas}
\label{algorithm:overall}
\end{algorithm2e}
\vspace{-\intextsep}

%
Our model generation technique  is described in \Cref{algorithm:overall}.
Function \code{\overall} takes as input a formula $\forall\set{x}.
\Phi\left(\set{x},\set{a}\right)$ over a theory $\T$.
It first calculates a $\sic_{\Phi,\set{x}}$ $\Psi\left(\set{a}\right)$ in
\qf-$\T$.
Then it solves $\Phi\left(\set{x},\set{a}\right) \wedge
\Psi\left(\set{a}\right)$.
Finally, depending on the result and whether $\Psi\left(\set{a}\right)$ is a
$\wic_{\Phi,\set{x}}$ or not, it answers \sat, \unsat\ or \unknown.
\code{\overall}  is parametrized by two functions \code{\solve} and
\code{\generic}:
\begin{itemize}[topsep=0pt]
  \item \code{\solve} is a decision procedure (typically a SMT solver) for
    \qf-$\T$.
    \code{\solve} is said to be \emph{correct} if each time it answers \sat\
    (resp.\ \unsat) the formula is satisfiable (resp.\ unsatisfiable); it is
    said to be \emph{complete} if it always answers \sat\ or \unsat, never
    \unknown.
    %
    %
  \item  \code{\generic} takes as input a formula $\Phi$ in \qf-$\T$ and a set
    of targeted  variables $\set{x}$, and produces a  $\sic_{\Phi,\set{x}}$ in
    \qf-$\T$.
    %
    It is said to be \emph{correct} if it always returns a \sic, and
    \emph{complete} if all the \sic\ it returns are \wic.
    A possible implementation of \code{\generic} is described in
    \Cref{algorithm:taint:generic} (\Cref{section:taint}).
\end{itemize}
%
%
Function \code{\overall} enjoys the two following properties, where correctness
and completeness are defined as for \code{\solve}.


\begin{theorem}[Correctness and completeness]
  \label{prop:algo_correct}
  \label{prop:algo_complet}
  \begin{itemize}[topsep=0pt]
    \item If \code{\solve} and \code{\generic} are  correct, then
      \code{\overall} is correct.
      %
      %
      %
    \item If \code{\solve} and \code{\generic} are complete, then
      \code{\overall} is complete.
      \end{itemize}
\end{theorem}

\begin{proof}[sketch of]
  Follow directly from \Cref{prop:sic_correct,prop:wic_complete} (\Cref{sec:SIC}).
\end{proof}

\subsection{Taint-based SIC inference}\label{section:taint}

\vspace{-\intextsep}
\begin{algorithm2e}
  \DontPrintSemicolon
  \SetKwFunction{FGeneric}{\generic}
  \SetKwFunction{FTheory}{\theory}
  \SetKwInput{KwDefault}{Default}
  \SetKwProg{KwPara}{Parameter:}{}{}
  \SetKwProg{Fn}{Function}{:}{}
  \SetKwProg{Either}{either}{}{}

  \KwPara{\FTheory}{
    \KwIn{$f$ a  function symbol,   its parameters $\phi_i$,  $\set{x}$ a set of
    targeted variables and  $\psi_i$ their associated $\sic_{\phi_i,\set{x}}$}
    \KwOut{$\Psi$ a $\sic_{f\left(\phi_i\right),\set{x}}$}
    \KwDefault{Return $\bot$}
  }\smallskip

  \Fn{\FGeneric{$\Phi$,$\set{x}$}}{
    \KwIn{$\Phi$ a formula  and $\set{x}$ a set of targeted variables}
    \KwOut{$\Psi$ a $\sic_{\Phi,\set{x}}$}
    \smallskip
    \lEither{$\Phi$ is a constant}{
      \KwRet{$\top$}
    }
    \lEither{$\Phi$ is a variable $v$}{
      \KwRet{$v\notin\set{x}$}
    }
    \Either{
      $\Phi$ is a function $f\left(\phi_1,.\,,\phi_n\right)$
    }{
      Let $\psi_i \triangleq \FGeneric\left(\phi_i,\set{x}\right)$ for all
      $i\in\left\{1,.\,,n\right\}$\\
      Let $\Psi_{\phantom{i}} \triangleq \FTheory\left(f, \left(\phi_1,\!.,\phi_n\right),
      \left(\psi_1,\!.,\psi_n\right), \set{x}\right)$\\
      \KwRet{$\Psi \vee {\bigwedge_i\psi_i}$}
    }}
\caption{Taint-based  \sic\ inference}
\label{algorithm:taint:generic}
\end{algorithm2e}
\vspace{-\intextsep}

%
%
\Cref{algorithm:taint:generic} presents a taint-based implementation of function
\code{\generic}.
It consists in a (syntactic) core calculus  described here, refined by a
(semantic) theory-dependent calculus \code{\theory} described in
\Cref{section:theory}.
%
%
%
%
From  formula $\Phi\left(\set{x},\set{a}\right)$ and targeted variables
$\set{x}$, \code{\generic} is defined recursively as follow.

%
If $\Phi$ is a constant
%
%
it returns $\top$ as constants are independent of any variable.
%
%
If $\Phi$ is a variable $v$, it returns $\top$ if we may depend on $v$ (i.e., $v
\not\in\set{x}$), $\bot$ otherwise.
%
If $\Phi$ is a function $f\left(\phi_1,.\,,\phi_n\right)$, it first recursively
computes for every sub-term $\phi_i$ a $\sic_{\phi_i,\set{x}}$ $\psi_i$.
%
Then these results are sent with $\Phi$ to \code{\theory} which computes a
$\sic_{\Phi,\set{x}}$ $\Psi$.
%
%
%
The procedure returns the disjunction between $\Psi$
and the conjunction of the $\psi_i$'s.
%
%
Note that \code{\theory} default value $\bot$ is absorbed by  the disjunction.

The intuition is that if the $\phi_i$'s are independent of $\set{x}$, then
$f\left(\phi_1,.\,,\phi_n\right)$ is.
Therefore \Cref{algorithm:taint:generic} is said to be \emph{taint-based} as,
when \code{\theory} is left to its default value, it acts as a form of taint
tracking \cite{DD77,O95} inside the formula.
%
%
%
%

\begin{proposition}[Correctness]\label{correctness}
  %
  %
  Given a formula $\Phi\left(\set{x},\set{a}\right)$ and assuming that
  \code{\theory} is correct, then $\code{\generic}\left(\Phi,\set{x}\right)$
  indeed computes a $\sic_{\Phi,\set{x}}$.
\end{proposition}

\begin{proof}[sketch of]
  This proof has been mechanized in \coq.
\end{proof}



Note that on the other hand, completeness does not hold: in general
$\code{\generic}$ does not compute a $\wic$, cf.\ discussion in
\Cref{discussion:wic}.

\subsection{Complexity and efficiency}\label{section:complexity}

We now evaluate the overhead induced by \Cref{algorithm:overall} in terms of
formula size and complexity of the resolution --- the running time of
\Cref{algorithm:overall} itself being expected to be negligible (preprocessing).
%
%


\begin{definition}
  The size of a term is inductively defined as
  $\mbox{size}\left(x\right)\triangleq 1$ for $x$ a variable, and
  $\mbox{size}\left(f\left(t_1,.\,,t_n\right)\right)\triangleq
  1+\Sigma_i\,\mbox{size}\left(t_i\right)$ otherwise.
  We say that \code{\theory} is bounded in size if there exists $K$ such that,
  for all term $\Delta$,\\
  ${\mbox{size} \left(\code{\theory}\left(\Delta,\cdot\right)\right) \leq {K}}$.
\end{definition}

\begin{proposition}[Size bound]\label{prop:size_bound}
  Let $N$ be the maximal arity of symbols defined by theory $\T$.
  %
  If \code{\theory} is bounded in size by $K$, then for all formula $\Phi$ in
  $\T$, $\mbox{size}\left(\code{\generic}\left(\Phi,\cdot\right)\right) \leq
  \left(K+N\right)\cdot\mbox{size}\left(\Phi\right)$.
\end{proposition}

\begin{proposition}[Complexity bound]\label{prop:complexity_bound}
  Let us suppose \code{\theory} bounded in size, and let $\Phi$ be a formula
  belonging to a theory $\T$ with polynomial-time checkable solutions.
  If $\Psi$ is a $\sic_{\Phi,\cdot}$ produced by \code{\generic}, then a
  solution for $\Phi\wedge\Psi$ is checkable in time polynomial in size of
  $\Phi$.
\end{proposition}

\begin{proof}[sketch of]
  \Cref{proof:size_bound,proof:complexity_bound}
\end{proof}

These propositions demonstrate that, for formula landing in complex enough
theories, our method lifts \qf-solvers to the quantified case (in an
approximated way) without any significant overhead, as long as \code{\theory} is
bounded in size.
%
%
%
This later constraint can be achieved by systematically binding sub-terms to
(constant-size) fresh names and having \code{\theory} manipulates these binders.
%

%
%
%




\subsection{Discussions}\label{section:extensions}

\paragraph{Extension}
%
%
Let us remarks that our framework encompasses partial quantifier elimination  as
long as the remaining quantifiers are handled by \code{\solve}.
%
For example, we may want to remove quantifications over arrays  but keep those
on bitvectors.
%
%
%
In this setting, \code{\generic} can also allow some level of quantification,
providing that \code{\solve} handles them.

\paragraph{About WIC} \label{discussion:wic}
%
%
%
%
%
%
As already stated, \code{\generic} does not propagate \wic\ in general.
%
%
For example, considering formulas $t_1 \triangleq (x<0)$ and $t_2 \triangleq
(x\geq 0)$, then $\wic_{t_1,x} = \bot$ and $\wic_{t_2,x} =  \bot$.
Hence  \code{\generic} returns $\bot$ as \sic\  for $t_1 \vee t_2$, while
actually $\wic_{t_1 \vee t_2,x} = \top$.
As a second example, let us consider term $f(a,x)$ with $f$ uninterpreted.
Then $\wic_{f(a,x),x} =  \top$ while \code{\generic} returns $\bot$\todo{le
prouver?}.
%

Nevertheless, we can already highlight a few  cases where \wic\ can be computed.
\code{\generic} does propagate \wic\ on one-to-one uninterpreted
functions\todo{le prouver?}.
If no variable $\set{x}$ is present in a sub-term, then the associated \wic\ is
$\top$.
While a priori naive, this case becomes interesting when combined with
simplifications (\Cref{sec:implementation}) that may eliminate $\set{x}$.
If a sub-term falls in a sub-theory admitting quantifier elimination, then the
associated \wic\ is computed by eliminating quantifiers in
$(\forall.\set{x}.\set{y}. \Phi(\set{x},\set{a}) \Leftrightarrow
\Phi(\set{y},\set{a}))$.
We may also think of dedicated patterns:
regarding bitvectors, the \wic\ for $x \leq a \wedge x \leq x+k$ is $a \leq
\code{Max}-k$.
\emph{Identifying under which condition \wic\ propagation holds is a strong
direction for future work.}



\section{Theory-dependent SIC refinements}\label{section:theory}

%
We now present  theory-dependent \sic\ refinements
%
for theories relevant to program analysis: booleans, fixed-size bitvectors and
arrays ---
recall that uninterpreted functions are already handled by
\Cref{algorithm:taint:generic}.
%
%
%
We then propose a generalization of these refinements together with a
correctness proof for a larger class of operators.
%

\subsection{Refinement on theories} \label{sec:theory-refinement}

%
We recall \code{\theory} takes four parameters:
a function symbol $f$,
%
its arguments  $\left(t_1,.\,,t_n\right)$,
%
%
%
%
%
%
their associated \sic\  $\left(\taint{t_1},.\,,\taint{t_n}\right)$,
%
and targeted variables $\set{x}$.
%
\code{\theory} pattern-matches the function symbol and returns the associated
\sic\ according to rules in \Cref{fig:taint}.
%
%
%
If a function symbol is not supported, we return  default value $\bot$.
Constants and variables are handled by \code{\generic}.
For the sake of simplicity, rules in \Cref{fig:taint} are defined recursively,
but can easily fit the interface required for \code{\theory} in
\Cref{algorithm:taint:generic} by turning recursive calls into parameters.
%

\paragraph{Booleans and \ite}
Rules for the boolean theory (\Cref{fig:bl_taint}) handles $\blimply$, $\bland$,
$\blor$ and $\ite$ (if-then-else).
For binary operators, the \sic\ is the conjunction of the \sic\ associated to
one of the two sub-terms and a constraint on this sub-term that forces the
result of the operator to be constant --- e.g., to be equal to $\bot$ (resp.\
$\top$) for the antecedent (resp.\ consequent) of an implication.
%
%
These equality constraints are based on absorbing elements of operators.

Inference for the $\ite$ operator is more subtle.
%
%
%
%
%
Intuitively, if its condition is independent of some $\set{x}$, we use it to
select the $\sic_{\set{x}}$ of the sub-term that will be selected by the $\ite$
operator.
If the condition is dependent of $\set{x}$, then we cannot use it anymore to
select a $\sic_{\set{x}}$.
In this case, we return the conjunction of the $\sic_{\set{x}}$ of both
sub-terms and the constraint that the two sub-terms are equal.

\begin{figure}[htbp]
  \makebox[\textwidth][c]{
  \subfloat[Booleans and \ite]{
    \begin{minipage}[b]{0.6\textwidth}
      \[
        \begin{array}{r@{\ }c@{\ }l}
          \taint{\left({a}\blimply{b}\right)} & \triangleq &
          \left(\taint{a}\wedge{a=\bot}\right)\vee
          \left(\taint{b}\wedge{b=\top}\right) \\
          \taint{\left({a}\bland{b}\right)} & \triangleq &
          \left(\taint{a}\wedge{a=\bot}\right)\vee
          \left(\taint{b}\wedge{b=\bot}\right) \\
          \taint{\left({a}\blor{b}\right)} & \triangleq &
          \left(\taint{a}\wedge{a=\top}\right)\vee
          \left(\taint{b}\wedge{b=\top}\right) \\
          \taint{\left(\ite\,c\,a\,b\right)} & \triangleq &
          \left(\taint{c}\wedge\ite\,c\,\taint{a}\,\taint{b}\right)\vee
          \left(\taint{a}\wedge\taint{b}\wedge{a=b}\right)
        \end{array}
      \]
    \end{minipage}
  \label{fig:bl_taint}}
  \subfloat[Fixed-size bitvectors]{
    \begin{minipage}[b]{0.6\textwidth}
      \[
        \begin{array}{r@{\ }c@{\ }l}
          \taint{\left({a_n}\bvand{b_n}\right)} & \triangleq &
          \left(\taint{a_n}\wedge{a_n=0_n}\right)\vee
          \left(\taint{b_n}\wedge{b_n=0_n}\right) \\
          \taint{\left({a_n}\bvor{b_n}\right)} & \triangleq &
          \left(\taint{a_n}\wedge{a_n=1_n}\right)\vee
          \left(\taint{b_n}\wedge{b_n=1_n}\right) \\
          \taint{\left({a_n}\bvmul{b_n}\right)} & \triangleq &
          \left(\taint{a_n}\wedge{a_n=0_n}\right)\vee
          \left(\taint{b_n}\wedge{b_n=0_n}\right) \\
          \taint{\left({a_n}\bvshl{b_n}\right)} & \triangleq &
          \left(\taint{b_n}\wedge{b_n \bvge n}\right) \\
        \end{array}
      \]
    \end{minipage}
  \label{fig:bv_taint}}}

  \makebox[\textwidth][c]{
  \subfloat[Arrays ]{
    \begin{minipage}[b]{2\textwidth}
      \[
        \begin{array}{r@{\ }c@{\ }l}
          \taint{\left(\select\,\left(\store\,a\,i\,e\right)\,j\right)} & \triangleq &
          \taint{\left(\ite\,\left(i=j\right)\,e\,\left(\select\,a\,j\right)\right)} \\
          & \triangleq & \left(\taint{\left(i=j\right)}\wedge\left(\ite\,\left(i=j\right)\,
          \taint{e}\,\taint{\left(\select\,a\,j\right)}\right)\right)
          \vee\left(\taint{e}\wedge\taint{\left(\select\,a\,j\right)}\wedge\left(e=\select\,a\,j\right)\right) \\
          & \triangleq & \left(\taint{i}\wedge\taint{j}
          \wedge\left(\ite\,\left(i=j\right)\,
          \taint{e}\,\taint{\left(\select\,a\,j\right)}\right)\right)
          \vee\left(\taint{e}\wedge\taint{\left(\select\,a\,j\right)}\wedge\left(e=\select\,a\,j\right)\right) \\
        \end{array}
      \]
    \end{minipage}
    \label{fig:ax_taint}}}
  \caption{Examples of  refinements for \code{\theory}}
\label{fig:taint}
\end{figure}


\paragraph{Bitvectors and arrays}
Rules for bitvectors (\Cref{fig:bv_taint}) follow similar ideas, with constant
$\top$ (resp.\ $\bot$) substituted by $1_n$ (resp.\ $0_n$), the bitvector of
size $n$ full of ones (resp.\ zeros).
%
%
%
%
Rules for arrays (\Cref{fig:ax_taint}) are derived from the theory axioms.
The definition is recursive: rules need be applied until reaching either a
\store\ at the position where the \select\ occurs, or the initial array
variable.

As a rule of thumb, good \sic\ can be derived from function axioms in the form
of rewriting rules, as done for arrays.
Similar constructions can be obtained for example for stacks or queues.





\subsection{$\R$-absorbing functions} \label{sec:taint:rabsorbant}

%
%
We propose a generalization of the previous theory-dependent \sic\ refinements
to a larger class of functions, and prove its correctness.

Intuitively, if a function has an absorbing element, constraining one of its
operands to be equal to this element will ensure that the result of the function
is independent of the other operands.
%
%
However, it is not enough when a relation between some elements is needed, such
as with $\left(t[a\leftarrow b]\right)[c]$ where constraint $a=c$ ensures the
independence with regards to $t$.
%
%
We thus generalize the notion of absorption to $\R$-absorption, where $\R$ is a
relation between function arguments.

\begin{definition}\label{def:r-absorbing-function}
  Let $f:\tau_1\times\cdots\times\tau_n\rightarrow\tau$ a function.
  $f$ is $\R$-absorbing if there exists
  $\I_\R\subset\left\lbrace1,\cdots,n\right\rbrace$ and $\R$ a relation between
  $\alpha_i:\tau_i,\,i\in\I_\R$ such that,
  %
  %
  %
  %
  for all $b \triangleq \left(b_1,\dots,b_n\right)$
  %
  %
  and $c\triangleq \left(c_1,\dots,c_n\right)\in\tau_1\times\dots\times\tau_n$,
  %
  %
  if $\R(\left.b\right|_{\I_{\R}})$ and $\left.b\right|_{\I_{\R}} =
  \left.c\right|_{\I_{\R}}$ where $\left.\cdot\right|_{\I_{\R}}$ is the
  projection  on $\I_{\R}$, then  $f(b)=f(c)$.
  %
  %
  %
  %

  $\I_\R$ is called the support of the relation of absorption $\R$.
\end{definition}

%
For example, $\left(a,b\right) \mapsto {{a}\vee{b}}$ has two pairs
$\left\langle\R,\,\I_\R\right\rangle$ coinciding with the usual notion of
absorption, $\left\langle{a\!=\!\top},\,\left\{1_a\right\}\right\rangle$ and
$\left\langle{b\!=\!\top},\,\left\{2_b\right\}\right\rangle$.
%
%
Function $\left(x,y,z\right) \mapsto {xy+z}$ has among others the pair
$\left\langle{x\!=\!0},\,\left\{1_x,3_z\right\}\right\rangle$, while
$\left(a,b,c,t\right) \mapsto \left(t[a \leftarrow b]\right)[c]$ has the pair
$\left\langle{a\!=\!c},\,\left\{1_a,3_c\right\}\right\rangle$.
%
%
We can now state the following proposition:

\begin{proposition}\label{prop:r-absorbing-function}
  Let $f\left(t_1,\dots,t_n\right)$ be a $\R$-absorbing function of support
  $\I_\R$, and let $\taint{t_i}$ be a $\sic_{t_i,\set{x}}$ for some $\set{x}$.
  Then $\R\left(t_{i\in\mathcal{\I_\R}}\right)
  \bigwedge_{i\in\mathcal{\I_\R}}\taint{t_i}$ is a $\sic_{f,\set{x}}$.
\end{proposition}

\begin{proof}[sketch of]
  \Cref{proof:r-absorbing}
\end{proof}


Previous examples (\Cref{sec:theory-refinement}) can be recast in term of
$\R$-absorbing function, proving their correctness (\Cref{fig:relation} in
\Cref{appendix:refinement-encoding}).
Note also that regarding our end-goal, we should accept only $\R$-absorbing
functions in \qf-$\T$.

%

\section{Experimental evaluation}\label{section:benchmarks}

This section describes the implementation of our method
(\Cref{sec:implementation}) for bitvectors and arrays (ABV), together with
experimental evaluation (\Cref{sec:evaluation}).
%
%

\subsection{Implementation}\label{sec:implementation}

Our prototype  \tfml\ (\emph{Taint engine for ForMuLa}) comprises 7 klocs of OCaml.
%
%
Given an input formula in the \smtlib\ format \cite{BST10} (ABV theory), \tfml\
performs several normalizations before adding taint information following
\Cref{algorithm:overall}.
%
%
The process ends with simplifications as taint usually introduces many constant
values, and a new \smtlib\ formula is output.

\paragraph{Sharing with let-binding}
This stage is crucial as it allows to avoid term duplication in \code{\theory}
(\Cref{algorithm:taint:generic,section:complexity,prop:size_bound}).
%
%
%
%
We introduce new names for relevant sub-terms in order to easily share them.
%
%

\paragraph{Simplifications}
%
%
%
%
We perform constant propagation and rewriting (standard rules, e.g.~$x-x \mapsto
0$ or $x \times 1 \mapsto x$) on both initial and transformed formulas --
equality is soundly approximated by syntactic equality.

\paragraph{Shadow arrays}
We encode taint constraints over arrays through \emph{shadow arrays}.
For each array declared in the formula, we declare a (taint) shadow  array.
The default value for all cells of the shadow array is the taint of the original
array,
and for each value stored (resp.\ read) in the original array, we store (resp.\
read) the taint of the value in the shadow array.
As logical arrays are infinite, we cannot constraint all the values contained in
the initial shadow array.
Instead, we rely on a common trick in array theory:
%
%
we constraint only cells corresponding to a relevant read index in the formula.
%
%
%

\paragraph{Iterative skolemization}
%
%
%
While we have supposed along the paper to work on skolemized formulas, we have
to be more careful in practice.
Indeed, skolemization introduce dependencies between a skolemized variable and
all its preceding universally quantified variables, blurring our analysis and
likely resulting in considering the whole formula as dependent.
%
%
%
Instead, we follow  an iterative process:
\begin{enumerate*}
  \item Skolemize the first block of existentially quantified variables;
  \item Compute the independence condition for any targeted variable in the
    first block of universal quantifiers and remove these quantifiers;
  \item Repeat.

\end{enumerate*}
This results in full Skolemization together with the construction of an
independence condition, while avoiding many unnecessary dependencies.


\subsection{Evaluation}\label{sec:evaluation}

\paragraph{Objective}
We experimentally evaluate the following research questions:
\begin{enumerate*}[label=\emph{RQ\arabic*}]
  \item \label{enum:comparison} How does our approach perform with regard to
    state-of-the-art approaches for model generation of quantified formulas?
  \item \label{enum:enhancement} How effective is it at lifting quantifier-free
    solvers into (\sat-only) quantified solvers?
  \item \label{enum:efficiency} How efficient is it in terms of preprocessing
time and formula size overhead?  \end{enumerate*}
We evaluate our method on a set of formulas combining arrays and bitvectors
(paramount in software verification),  against state-of-the-art solvers for
these theories.

\paragraph{Protocol}
The experimental setup below runs on an Intel(R) Xeon(R) E5-2660 v3 @ 2.60GHz,
4GB RAM per process, and a \timeout\ of 1000s per formula.
\begin{description}[topsep=0pt, leftmargin=.4cm]
  \item[Metrics]
    For \ref{enum:comparison} we compare the number of \sat\  and \unknown\
    answers  between solvers supporting  quantification, with and without our
    approach.
    For \ref{enum:enhancement}, we compare the number of \sat\ and  \unknown\
    answers between quantifier-free solvers enhanced by our approach and solvers
    supporting quantification.
    For \ref{enum:efficiency}, we measure  preprocessing time and formulas size
    overhead.
    %
    %
  \item[Benchmarks]
    We consider two sets of \abv\ formulas.
    First, a set of 1421 formulas from (a modified version of) the symbolic
    execution tool \binsec~\cite{DBTMFPM16} representing quantified reachability
    queries (cf.~\Cref{sec:motivation}) over  \binsec\ benchmark programs
    (security challenges, e.g.~\code{crackme} or vulnerability finding).
    The initial (array) memory is quantified so that models depend only on user
    input.
    %
    %
    %
    %
    Second, a set of 1269 \abv\ formulas generated from formulas of the
    \qf-\abv\ category of \smtlib\ \cite{BST10} --
    %
    %
    %
    sub-categories \code{brummayerbiere}, \code{dwp\;formulas} and
    \code{klee\;selected}.
    %
    %
    The generation process consists in universally quantifying some of the
    initial array variables, mimicking quantified reachability problems.
    %
    %
    %
    %
  \item[Competitors]
    For \ref{enum:comparison}, we compete against the two  state-of-the-art SMT
    solvers for quantified formulas  CVC4  \cite{BCDHJKRT11} (finite model
    instantiation \cite{RTGK13}) and Z3 \cite{MB08} (mode-based instantiation
    \cite{GM09}).
    %
    We also consider degraded versions CVC4$_{E}$ and Z3$_{E}$ that roughly
    represent standard E-matching \cite{DNS05}.
    For \ref{enum:enhancement}  we use Boolector~\cite{BB09}, one of the very
    best \qf-ABV solvers.
    %
\end{description}

\begin{table}
  \caption{Answers and resolution time (in seconds, include \timeout)}
  \label{table:all-solvers-smt-lib}
  \label{table:resolution_time}
  \makebox[\textwidth][c]{
    \begin{tabular}{|c|c|c|c|c|c|c|c|c|c|c|c|c|c|}\cline{4-12}
      \multicolumn{3}{c|}{\strut} &
      Boolector\raisebox{1pt}{$\bullet$} &
      CVC4 & CVC4\raisebox{1pt}{$\bullet$} &
      CVC4$_{E}$ & CVC4$_{E}$\raisebox{1pt}{$\bullet$} &
      Z3 & Z3\raisebox{1pt}{$\bullet$} &
      Z3$_{E}$ & Z3$_{E}$\raisebox{1pt}{$\bullet$} \\\hline

      \multirow{4}{*}{\rotatebox{90}{\smtlib\strut}}
      &    & \sat\strut     & \bf 399 & 84   & 242  & 84   & 242  & 261  & 366  & 87   & 366  \\
      & \# & \unsat\strut   & N/A     & 0    & N/A  & 0    & N/A  & 165  & N/A  & 0    & N/A  \\
      &    & \unknown\strut & 870     & 1185 & 1027 & 1185 & 1027 & 843  & 903  & 1182 & 903  \\\cline{2-12}
      & \multicolumn{2}{c|}{total time\strut}
               & \bf 349 & 165 & 194\,667 & 165 & 196\,934 & 270\,150 & 36\,480 & 192 &  41\,935 \\\hline\hline

      \multirow{4}{*}{\rotatebox{90}{\binsec\strut}}
      &    & \sat\strut     & \bf 1042& 951  & 954  & 951  & 954  & 953  & \bf 1042& 953  & \bf 1042\\
      & \# & \unsat\strut   & N/A     & 62   & N/A  & 62   & N/A  & 319  & N/A     & 62   & N/A     \\
      &    & \unknown\strut & 379     & 408  & 467  & 408  & 467  & 149  & 379     & 406  & 379     \\\cline{2-12}
      & \multicolumn{2}{c|}{total time\strut}
               & \bf 1152 & 64\,761 & 76\,811 & 64\,772 & 77\,009 & 30\,235 &  11\,415 & 135  & 11\,604 \\\hline

      \multicolumn{12}{c}{
        \begin{tabular}{c@{\qquad}c}
          solver\raisebox{1pt}{$\bullet$}: solver enhanced with our method &
          Z3$_{E}$, CVC4$_{{E}}$: essentially E-matching\\
        \end{tabular}} \\
    \end{tabular}}
\end{table}

\paragraph{Results}
\Cref{table:all-solvers-smt-lib,table:complementarity} and  \Cref{fig:overhead}
sum up our experimental results, which  have all  been cross-checked for
consistency.
\Cref{table:all-solvers-smt-lib} reports the number of successes (\sat\ or
\unsat) and failures (\unknown), plus total solving times (average time in
\Cref{appendix:sec:benchs}).
The \raisebox{1pt}{$\bullet$} sign indicates formulas preprocessed with our
approach.
In that case it is impossible to correctly answer \unsat\ (no \wic\ checking),
the \unsat\ line is thus N/A.
Since Boolector  do not support quantified \abv\ formulas, we only give results
with our approach enabled.
\Cref{table:all-solvers-smt-lib} reads as follow: of the 1269 \smtlib\ formulas,
standalone Z3 solves 426 formulas (261 \sat, 165 \unsat), and 366 (all \sat) if
preprocessed.
%
%
Interestingly, our approach always improves the underlying solver in terms of
solved (\sat) instances, either in a significant way (\smtlib) or in a modest
way (\binsec). Yet, recall that in a software verification setting every win
matters (possibly new bug found or new assertion proved).
For Z3\raisebox{1pt}{$\bullet$}, it also strongly reduces computation time.
Last but not least, Boolector\raisebox{1pt}{$\bullet$} (a pure \qf-solver)
turns out to have the best performance on \sat-instances, beating
state-of-the-art approaches both in terms of solved instances and computation
time.

\Cref{table:complementarity}  substantiates the complementarity of the different
methods, and reads as follow: for \smtlib, Boolector\raisebox{1pt}{$\bullet$}
solves 224 (\sat) formulas missed by Z3, while Z3 solves 86 (\sat) formulas
missed by Boolector\raisebox{1pt}{$\bullet$}, and 485 (\sat) formulas are solved
by either one of them.

%
%
%
%
%
%
%
%

\Cref{fig:overhead} shows formula size averaging a 9-fold increase (min 3, max
12): yet they are easier to solve because more constrained.
Regarding performance and overhead of  the tainting process, \emph{taint time is
almost always less than 1s} in our experiments (not shown here), 4min for worst
case, clearly dominated by resolution time.
The worst case is due to a pass of linearithmic complexity which can be
optimized to be logarithmic.

\begin{table}
  \caption{Complementarity of our approach with existing solvers (\sat\ instances)}
  \label{table:complementarity}

  \makebox[\textwidth][c]{
    \begin{tabular}{|c|c|ccc|ccc|ccc|}\cline{3-11}
      \multicolumn{2}{c|}{} &
      \multicolumn{3}{c|}{CVC4\raisebox{1pt}{$\bullet$}} &
      \multicolumn{3}{c|}{Z3\raisebox{1pt}{$\bullet$}} &
      \multicolumn{3}{c|}{Boolector\raisebox{1pt}{$\bullet$}} \\\hline
      \multirow{2}{*}{\smtlib\strut}
      & CVC4 & -10 & +168 & [252] &      &      &       & -10 & +325 & [409] \\\cline{2-11}
      & Z3   &     &      &       & -119 & +224 & [485] & -86 & +224 & [485] \\\hline\hline
      \multirow{2}{*}{\binsec\strut}
      & CVC4 & -25 & +28 & [979] &      &      &        & -25 & +116 & [1067] \\\cline{2-11}
      & Z3   &     &      &       & -25 & +114 & [1067] & -25 & +114 & [1067] \\\hline
    \end{tabular}}
\end{table}

\paragraph{Pearls} We show hereafter two particular applications of our method.
\Cref{tab:grub} reports results of another symbolic execution experiment, on the
\code{grub} example.
On this example, Boolector\raisebox{1pt}{$\bullet$} completely outperforms
existing approaches.
As a second application, while the main drawback of our method is that it
precludes proving \unsat, this is easily mitigated by complementing the approach
with another one geared at (or able to) proving \unsat, yielding efficient
solvers for quantified formulas, as shown in \Cref{tab:combination}.

\paragraph{Conclusion}
Experiments demonstrate the relevance of our taint-based technique for model
generation.
(\ref{enum:comparison}) Results in Table \ref{table:all-solvers-smt-lib} shows
that our approach greatly facilitates the resolution process.
\emph{On these examples, our method performs better than state-of-the-art
solvers but also strongly complements them (\Cref{table:complementarity})}.
%
%
%
%
%
%
(\ref{enum:enhancement}) Moreover, \Cref{table:all-solvers-smt-lib} demonstrates
that our technique is highly effective at  lifting quantifier-free solvers to
quantified formulas, in both number of \sat\ answers
and computation time.
\emph{Indeed, once lifted, Boolector perform better (for \sat-only) than Z3 or
CVC4 with full quantifier support}.
%
%
Finally (\ref{enum:efficiency}) our tainting method itself is very efficient
both in time and space, making it perfect either for a preprocessing step or for
a deeper integration into a solver.
In our current prototype implementation, we consider the cost to be low.


%
%
%
%
%
%
%
%
%
%

\begin{minipage}[t]{0.95\textwidth}
  \begin{minipage}{0.49\textwidth}
    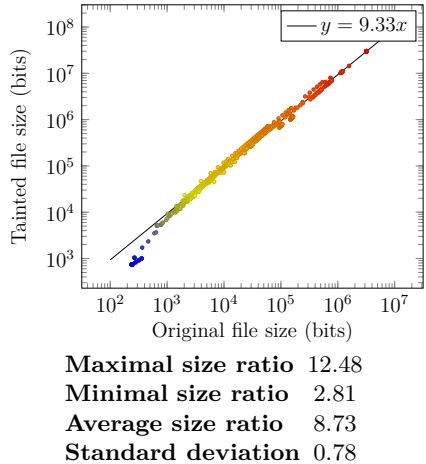
\begin{figure}[H]
      \begin{minipage}{\textwidth}
        \large
        \resizebox{\textwidth}{!}{
          \begin{tikzpicture}
            \begin{axis}[
                xmode=log,
              ymode=log,
              xlabel={Original file size (bits)},
              ylabel={Tainted file size (bits)},
              scaled ticks=false,
            ]
              \addplot[samples=10, domain=100:10000000] {9.3344672941*x};
              \addlegendentry{$y=9.33x$};
              \addplot[only marks,scatter,mark size=1pt] table{size.txt};
            \end{axis}
          \end{tikzpicture}}
      \end{minipage}\\
      \begin{minipage}{\textwidth}
        \makebox[\textwidth][c]{
          \begin{tabular}{lc}
            \bf Maximal size ratio & 12.48  \\
            \bf Minimal size ratio & 2.81   \\
            \bf Average size ratio & 8.73   \\
            \bf Standard deviation & 0.78   \\
          \end{tabular}}
      \end{minipage}
      \caption{Overhead in formula size}
      \label{fig:overhead}
    \end{figure}
  \end{minipage}
  \begin{minipage}{0.49\textwidth}
    \begin{minipage}{\textwidth}
      \begin{table}[H]
        \caption{GRUB example} \label{tab:grub}
        \makebox[\textwidth][c]{
          \begin{tabular}{|c|c|c|c|}\cline{3-4}
            \multicolumn{2}{c|}{\strut} & Boolector\raisebox{1pt}{$\bullet$} & Z3 \\\hline
            & \sat\strut     & \bf 540 &   1 \\\cline{2-4}
            \# & \unsat\strut   & N/A &  42 \\\cline{2-4}
            & \unknown\strut & 355 & 852 \\\hline
            \multicolumn{2}{|c|}{total time} & \bf 16\,732 & 159\,765 \\\hline
          \end{tabular}}
      \end{table}
    \end{minipage}
      \begin{minipage}{\textwidth}
        \begin{table}[H]
          \caption{Best approaches} \label{tab:combination} \footnotesize

          \makebox[\textwidth][c]{
            \begin{tabular}{|c|c|c@{\ \vrule width 1pt\ }c|c|}\cline{3-5}
              \multicolumn{2}{c|}{} &
              \multicolumn{1}{c@{\ \vrule width 1pt\ }}{former} &
              \multicolumn{2}{c|}{new}\\\cline{3-5}
              \multicolumn{2}{c|}{} & Z3
              & B\raisebox{1pt}{$\bullet$}
              & B\raisebox{1pt}{$\bullet$} $\triangleright$ Z3 \\\hline

              \multirow{4}{*}{\rotatebox{90}{\smtlib\strut}}
              & \sat\strut     & 261      & 399    & 485   \\
              & \unsat\strut   & 165      & N/A    & 165   \\
              & \unknown\strut & 843      & 870    & 619   \\\cline{2-5}
              & time           & 270\,150 & 350    & 94\,610 \\\hline\hline

              \multirow{4}{*}{\rotatebox{90}{\binsec\strut}}
              & \sat\strut     & 953     & 1042     & 1067     \\
              & \unsat\strut   & 319     & N/A      & 319      \\
              & \unknown\strut & 149     & 379      & 35       \\\cline{2-5}
              & time           & 64\,761 & 1\,152 & 1\,169 \\\hline
            \end{tabular}}
        \end{table}
      \end{minipage}
  \end{minipage}
\end{minipage}

\section{Related work}\label{sec:related}

\todo{la ref pointée par un des reviewer : Model Finding for Recursive Functions
in SMT. IJCAR 2016.  NDSB : certes. Et je dis quoi dessus ?}

Traditional approaches to solving quantified formulas essentially involve either
{generic methods geared at proving unsatisfiability} and validity \cite{DNS05},
or complete but { dedicated approaches} for particular theories
\cite{BMS06,WHM10}.
%
%
Besides, some recent methods \cite{IJS08,GM09,RTGK13}  aim to be correct and
complete for larger classes of theories.

\paragraph{Generic method for unsatisfiability}
Broadly speaking, these methods iteratively instantiate axioms until a
contradiction is found.
They are generic w.r.t.~the underlying theory and allow to reuse standard theory
solvers, but termination is not guaranteed.
Also, they are more suited to prove unsatisfiability than to find models.
In this family, E-matching \cite{DNS05,MB07} shows reasonable cost when combined
with conflict-based instantiation \cite{RTM14} or semantic triggers
\cite{DCKP12,DCKP16}.
In pure first-order logic (without theories), quantifiers are mainly handled
through resolution and superposition \cite{BG94,NR01} as done in Vampire
\cite{RV02,KV13} and E \cite{S02}.

\paragraph{Complete methods for specific theories}
Much work has been done on designing complete decision procedures for quantified
theories of interest, notably {array properties} \cite{BMS06}, {quantified
theory of bitvectors} \cite{WHM10}, Presburger arithmetic or Real Linear
Arithmetic \cite{BKRW11}.
Yet, They usually come at a high cost.

\paragraph{Generic methods for model generation}
Some recent works detail attempts at more general approaches to model
generation.
%

\emph{Local theory extensions} \cite{IJS08,BRKBW15} provide means to extend some
decidable theories with free symbols and quantifications, retaining
decidability.
The approach identifies specific forms of formulas and quantifications
(bounded), such that these theory extensions can be solved using finite
instantiation of quantifiers together with a decision procedure for the original
theory.
The main drawback is that the formula size can highly increase.

\emph{Model-based quantifier instantiation} is an active area of research
notably developed in Z3 and CVC4.
The basic line is to consider the partial model under construction in order to
find the right quantifier instantiations, typically in a try-and-refine manner.
Depending on the variants, these methods favors either satisfiability or
unsatisfiability.
They build on the underlying quantifier-free solver and can be mixed with
E-matching techniques, yet each refinement yields a solver call and the
refinement process may not terminate.
Ge and de Moura \cite{GM09} study decidable fragments of first-order logic
modulo theories for which model-based quantifier instantiation yields soundness
and refutational completeness. 
%
%
Reynolds et al. \cite{RTM14}, in CVC4, and Barbosa \cite{B16}, in veriT, use
models to guide the instantiation process towards instances refuting the current
model.
%
%
\emph{Finite model quantifier instantiation} \cite{RTGK13,RTGKDB13} reduces the
search to finite models, and is indeed geared toward model generation rather
than unsatisfiability.
Similar techniques have been used in program synthesis \cite{RDKTB15}.

We drop support for the unsatisfiable case but get more flexibility:
we deal with quantifiers on any sort, the approach terminates and is
lightweight, in the sense that it requires a single call to the underlying
quantifier-free solver.
%

\paragraph{Tainting and non-interference}
Our method can be seen as taking inspiration from program taint analysis
\cite{DD77,O95} developed for checking the non-interference \cite{S07} of public
and secrete input in security-sensitive programs.
As far as the analogy goes, our approach should not be seen as checking
non-interference, but rather as inferring preconditions of non-interference.
Moreover, our formula-tainting technique is closer to dynamic program-tainting
than to static program-tainting, in the sense that precise dependency conditions
are statically inserted at preprocess-time, then precisely explored at
solving-time.

\section{Conclusion}

This paper addresses the problem of generating models of quantified first-order
formulas over built-in theories.
%
We propose a correct and generic approach based on a reduction to the
quantifier-free case through the inference of independence conditions.
The technique is applicable to any theory with a decidable quantifier-free case
and allows to reuse all the work done on quantifier-free solvers.
The method significantly enhances the performances of state-of-the-art SMT
solvers for the quantified case, and supplements the latest advances in the
field.


Future developments aim to tackle the definition of more precise inference
mechanisms of independence conditions, the identification of interesting
subclasses for which inferring weakest independence conditions is feasible, and
the combination with other quantifier instantiation techniques.

\newpage

\bibliographystyle{abbrv}
\bibliography{biblio}

\newpage

\appendix

{ \em \noindent Additional  technical details  are provided here for  reviewer convenience.  This  content will be   available
in an online technical report. }

%
%
%
%
%
%

\section{Background: detailed notations}

This work uses many-sorted first-order logic with equality.
Let $\set\S$ be the set of \emph{sort symbols}, and for every $\S\in\set\S$ let
$\set\X_\S$ be the set of \emph{variables} of sort $\S$.
We assume the sets $\set\X_\S$ to be pairwise disjoint and note $\set\X$ their
union.

A \emph{signature} $\Sigma$ consists of a set $\Sigma^s\subseteq\set\S$ of sort
symbols and a set $\Sigma^f$ of \emph{sorted functions symbols}
$f^{\S_1,\dots,\S_n,\S}$, where ${n}\geq{0}$ and
$\S_1,\dots,\S_n,\S\in\Sigma^s$.
Given a signature $\Sigma$, well-sorted terms and (possibly quantified) formulas
with variables in $\set\X$ are inductively defined as usual.
We refer to them respectively as \emph{$\Sigma$-terms} and
\emph{$\Sigma$-formulas}.
A \emph{ground term} (resp.\ \emph{formula}) is a $\Sigma$-term (resp.\ formula)
without variables.

Given a tuple of variables $\set{x} \triangleq \left(x_1,\dots,x_n\right)$ and a
quantifier $\Q$ ($\forall$ or $\exists$), we shorten $\Q{x_1}\dots\Q{x_n}.\Phi$
as $\Q\set{x}.\Phi$ .
A formula in \emph{prenex normal form} has all its quantifiers at head position
and ranging over the whole quantifier-free formula, so that it is written as
$\Q_1{\set{x}_1}\dots\Q_n{\set{x}_n}.\Phi$ where $\Phi$ is a quantifier-free
formula.
A formula is in \emph{Skolem normal form} if it is in prenex normal form with
only universal quantifiers.
If $\Phi$ is a $\Sigma$-formula, we write $\Phi\left(\set{x}\right)$ to denote
that the free variables of $\Phi$ are in $\set{x}$.
Let $\set{t}\triangleq\left(t_1,\dots,t_n\right)$ be a term tuple, we write
$\Phi\left(\set{t}\right)$ for the formula obtained from $\Phi$ by
simultaneously replacing each occurrence of $x_i$ in $\Phi$ by $t_i$.

A \emph{$\Sigma$-interpretation} $\I$ maps each $\S\in\Sigma^s$ to a non-empty
set $\S^\I$, the \emph{domain} of $\S$ in $\I$, each $x\in\set\X$ of sort $\S$
to an element $x^\I\in\S^\I$, and each $f^{\S_1,\dots,\S_n,\S}\in\Sigma^f$ to a
total function $f^\I:{\S_1^\I}\times\dots\times{\S_n^\I}\rightarrow{\S^\I}$.
A satisfiability relation $\models$  between $\Sigma$-interpretations and
$\Sigma$-formulas is defined inductively as usual.
We sometimes refer to models as ``solutions''.

A \emph{theory} is a pair $\T\triangleq\left(\Sigma,\set\I\right)$ where
$\Sigma$ is a signature and $\set\I$ a class of $\Sigma$-interpretations that is
closed under variable reassignment and isomorphism.
A formula $\Phi\left(\set{x}\right)$ of $\T$ is \emph{valid in $\T$} (resp.\
\emph{satisfiable}, resp.\ \emph{unsatisfiable}) if it is satisfied by all
(resp.\ some, resp.\ no) interpretations $\I\in\set\I$.
A set $\Gamma$ of formulas \emph{entails in $\T$} a $\Sigma$-formula $\Phi$,
written $\Gamma\models_\T\Phi$, if every interpretation in $\set\I$ that
satisfies all formulas in $\Gamma$ satisfies $\Phi$ as well.

%
%

\newpage
\section{The \abv\ theory of \smtlib}

\subsection{Definition}

Definition of boolean, fixed-size bitvector and array theories studied in
\Cref{sec:theory-refinement}.

\vspace{-\intextsep}
\begin{figure}

  \makebox[\textwidth][c]{
    \subfloat[Booleans and \ite]{
      \begin{minipage}[b]{0.4\textwidth}
        \[
          \begin{array}{r@{\ :\ }l}
            \top,\bot & \smtbl \\
            \blnot & \smtbl\rightarrow\smtbl \\
            \blimply,\bland,\blor,\blxor & \smtbl\rightarrow\smtbl\rightarrow\smtbl \\
            = & \A\rightarrow\A\rightarrow\smtbl \\
            \ite & \smtbl\rightarrow\A\rightarrow\A\rightarrow\A \\
          \end{array}\\
        \]
      \end{minipage}
  \label{bl_theory}}

  \subfloat[Fixed-size bitvectors]{
    \begin{minipage}[b]{0.9\textwidth}
      \[
        \begin{array}{r@{\ :\ }l@{\ \mbox{with}\ }l}
          \concat & \smtbv{i}\rightarrow\smtbv{j}\rightarrow\smtbv{(i+j)} & i>0,\,j>0 \\
            \extract\,i\,j & \smtbv{m}\rightarrow\smtbv{(i-j+1)} & m > i \geq j \geq 0 \\
            {U} & \smtbv{m}\rightarrow\smtbv{m} & m>0 \\
            {B} & \smtbv{m}\rightarrow\smtbv{m}\rightarrow\smtbv{m} & m>0 \\
            {C} & \smtbv{m}\rightarrow\smtbv{m}\rightarrow\smtbl & m>0 \\[.2cm]
            \multicolumn{3}{c}{
            {U}=\left\{\bvnot,\dots\right\},
            {B}=\left\{\bvor,\bvand,\bvadd,\bvsub,\bvmul,\bvdiv,\bvshl,\bvshr,\dots\right\},
            {C}=\left\{\bvlt,\bvgt,\bvle,\bvge,\dots\right\}}\\
        \end{array}
        \]
    \end{minipage}
    \label{bv_theory}}}

  \makebox[\textwidth][c]{
    \subfloat[Arrays]{
      \begin{minipage}[b]{0.6\textwidth}
        \[\begin{array}{@{}r@{\,:\,}l@{}}
          \select & \smtar{\A}{\B}\rightarrow\A\rightarrow\B \\
          \store  & \smtar{\A}{\B}\rightarrow\A\rightarrow\B\rightarrow\smtar{\A}{\B} \\
        \end{array}\]
      \end{minipage}
      \begin{minipage}[b]{0.6\textwidth}
        \[\begin{array}{@{}l@{}}
          \forall{a\,i\,e}.\,\select\left(\store\ a\ i\ e\right)i=e \\
          \forall{a\,i\,j\,e}.\left(i \neq j\right)
          \Rightarrow\select\left(\store\ a\ i\ e\right)j=\select\ a\ j \\
          \forall{a\,b}.\left(\forall{i}.\,\select\ a\ i=\select\ b\ i\right)
          \Rightarrow{a=b} \\
        \end{array}\]
      \end{minipage}
    \label{ax_theory}}
  \caption{Definition of the \abv\ theory}}

  \label{appendix:theories}

\end{figure}

\subsection{Relation of absorption}

\Cref{fig:taint} recast in term of of $\R$-absorbing function, as mention in
\Cref{sec:taint:rabsorbant}.

\vspace{-\intextsep}
\label{appendix:refinement-encoding}
\begin{figure}
  \subfloat[Booleans]{
    \begin{minipage}{0.5\textwidth}
      \[
        \begin{array}{r@{\ }c@{\ }l}
          {a}\blimply{b} & : &
          \left\langle{a\!=\!\bot},\,\left\{1_a\right\}\right\rangle,
          \left\langle{b\!=\!\top},\,\left\{2_b\right\}\right\rangle \\
          {a}\bland{b} & : &
          \left\langle{a\!=\!\bot},\,\left\{1_a\right\}\right\rangle,
          \left\langle{b\!=\!\bot},\,\left\{2_b\right\}\right\rangle \\
          {a}\blor{b} & : &
          \left\langle{a\!=\!\top},\,\left\{1_a\right\}\right\rangle,
          \left\langle{b\!=\!\top},\,\left\{2_b\right\}\right\rangle \\
        \end{array}
      \]
    \end{minipage}
  \label{fig:bl_relation}}
\subfloat[Fixed-size bitvectors]{
  \begin{minipage}{0.5\textwidth}
    \[
      \begin{array}{r@{\ }c@{\ }l}
        {a_n}\bvor{b_n} & : &
          \left\langle{a_n\!=\!1_n},\,\left\{1_a\right\}\right\rangle,
          \left\langle{b_n\!=\!1_n},\,\left\{2_b\right\}\right\rangle \\
          {a_n}\bvand{b_n} & : &
          \left\langle{a_n\!=\!0_n},\,\left\{1_a\right\}\right\rangle,
          \left\langle{b_n\!=\!0_n},\,\left\{2_b\right\}\right\rangle \\
          {a_n}\bvmul{b_n} & : &
          \left\langle{a_n\!=\!0_n},\,\left\{1_a\right\}\right\rangle,
          \left\langle{b_n\!=\!0_n},\,\left\{2_b\right\}\right\rangle \\
      \end{array}
      \]
  \end{minipage}
  \label{fig:bv_relation}}
\\\subfloat[if then else]{
  \begin{minipage}{\textwidth}
    \[
      \begin{array}{r@{\ }c@{\ }l}
        \ite\,c\,a\,b & : &
          \left\langle{c\!=\!\top},\,\left\{1_c,2_a\right\}\right\rangle,
          \left\langle{c\!=\!\bot},\,\left\{1_c,3_b\right\}\right\rangle,
          \left\langle{a\!=\!b},\,\left\{2_a,3_b\right\}\right\rangle \\
      \end{array}
      \]
  \end{minipage}
  \label{fig:ite_relation}}
\caption{Relation of absorption for some \abv\ function symbols.}
  \label{fig:relation}
\end{figure}

\newpage
\section{Proofs}

\subsection{Model generalization (\Cref{prop:sic_correct})}\label{proof:sic_correct}

\emph{%
Let $\Phi\left(\set{x},\set{a}\right)$  a formula and $\Psi$  a
$\sic_{\Phi,\set{x}}$.
If there exists an interpretation $\left\lbrace\val{x},\val{a}\right\rbrace$
such that $\left\lbrace\val{x},\val{a}\right\rbrace \models
\Psi\left(\set{a}\right) \wedge \Phi\left(\set{x},\set{a}\right)$, then
$\left\lbrace\val{a}\right\rbrace \models
\forall\set{x}.\Phi\left(\set{x},\set{a}\right)$.
}

\begin{proof}
  Let $(\val{x},\val{a})$ an interpretation of
  $\Phi\left(\set{x},\set{a}\right)$, and let us assume that  $(\val{x},\val{a})
  \models \Psi\left(\set{a}\right) \wedge \Phi\left(\set{x},\set{a}\right)$.
  It comes immediately that  $(\val{x},\val{a}) \models
  \Psi\left(\set{a}\right)$ (E1) and $(\val{x},\val{a}) \models
  \Phi\left(\set{x},\set{a}\right)$ (E2).  From (E1) we deduce that   $\val{a}
  \models \Psi\left(\set{a}\right)$, since $\set{x}$ does not appear in $\Psi$.
  Then, $\Psi$ being a $\sic_{\Phi,\set{x}}$, we get that
  $({\forall\set{x}.\forall\set{y}.  \Phi\left(\set{x},\val{a}\right)
  \Leftrightarrow \Phi\left(\set{y},\val{a}\right)})$ holds true (E3).
  Combining (E3) with the fact that $\Phi\left(\val{x},\val{a}\right)$ holds
  true (from E2), we conclude that  $\Phi\left(\set{x},\val{a}\right)$ is
  satisfied for any value of $\set{x}$, hence $\val{a} \models
  \forall\set{x}.\Phi\left(\set{x},\set{a}\right)$.
\end{proof}

\subsection{Model specialization (\Cref{prop:wic_complete})}\label{proof:wic_complete}

\emph{%
Let $\Phi\left(\set{x},\set{a}\right)$  a formula and $\Pi(\set{a})$  a
$\wic_{\Phi,\set{x}}$.
If there exists an interpretation $\left\lbrace\val{a}\right\rbrace$ such that
$\left\lbrace\val{a}\right\rbrace \models \forall\set{x}.
\Phi\left(\set{x},\set{a}\right)$, then
$\left\lbrace\val{x},\val{a}\right\rbrace \models \Pi\left(\set{a}\right) \wedge
\Phi\left(\set{x},\set{a}\right)$ for any valuation $\val{x}$  of $\set{x}$.
}

\begin{proof}
  Let $\val{a}$  an interpretation such that $\val{a} \models
  \forall\set{x}.\Phi\left(\set{x},\set{a}\right)$ (E1).
  Then by definition, $\forall\set{x}.\Phi\left(\set{x},\val{a}\right)$ is true.
  Especially $\forall\set{x}.\forall\set{y}.
  \Phi\left(\set{x},\val{a}\right) \Leftrightarrow
  \Phi\left(\set{x},\val{a}\right)$ also is.
  Therefore $\val{a} \in \sem{\Pi(\set{a})}$, another way of stating that
  $\val{a} \models \Pi(\set{a})$ (E2).
  Let us now consider an arbitrary value $\val{x}$ for $\set{x}$.
  By combining (E2) and (E1), we obtain that $(\val{x},\val{a}) \models
  \Pi\left(\set{a}\right) \wedge \Phi\left(\set{x},\set{a}\right)$.
\end{proof}

\subsection{Size bound (\Cref{prop:size_bound})}\label{proof:size_bound}

\emph{%
Let $N$ be the maximal arity of symbols defined by theory $\T$.
%
%
If \code{\theory} is bounded in size by $K$, then for all formula $\Phi$ in
$\T$, $\mbox{size}\left(\code{\generic}\left(\Phi,\cdot\right)\right) \leq
\left(K+N\right)\cdot\mbox{size}\left(\Phi\right)$.
}

\begin{proof}
  Let $\Phi \triangleq f\left(\phi_1,\!.,\phi_n\right)$ be a formula, let
  $\psi_i \triangleq \code{\generic}\left(\phi_i\right)$ be results of recursive
  calls, and let $\Psi \triangleq \code{\theory}\left(f,
  \left(\phi_1,\!.,\phi_n\right), \left(\psi_1,\!.,\psi_n\right), \set{x}\right)$.
  \[\begin{array}{r@{\ }c@{\ }c@{\ +\ }c@{\ +\ }c}
    \mbox{size}\left(\code{\generic}\left(f\left(\phi_1,\dots,\phi_n\right)\right)\right)
    & = & \multicolumn{3}{@{}l}{\mbox{size}\left(\Psi \vee {\bigwedge_i\psi_i}\right)} \\
    & = & \mbox{size}\left(\Psi\right) & 1+(n-1) & \Sigma_i\mbox{size}\left(\psi_i\right) \\
    & \leq & K & N & \Sigma_i\mbox{size}\left(\psi_i\right) \\
  \end{array}
  \]
  Then by structural induction we have
  $\mbox{size}\left(\code{\generic}\left(\Phi,\cdot\right)\right) \leq
  {\left(K+N\right)\cdot\mbox{size}\left(\Phi\right)}$.
\end{proof}

\subsection{Complexity bound (\Cref{prop:complexity_bound})}\label{proof:complexity_bound}

\emph{%
Let us suppose \code{\theory} bounded in size, and let $\Phi$ be a formula
belonging to a theory $\T$ with polynomial-time checkable solutions.
If $\Psi$ is a $\sic_{\Phi,\cdot}$ produced by \code{\generic}, then a solution
for $\Phi\wedge\Psi$ is checkable in time polynomial in size of $\Phi$.
}

\begin{proof}
  If \code{\theory} is bounded in size, then \code{\generic} is linear in size
  by Proposition \ref{prop:size_bound}.
  So for any formula $\Phi$ in a theory $\T$, if $\Psi$ is the \sic\ produced by
  \code{\generic}, then $\Phi\wedge\Psi$ is linearly proportional to $\Phi$.
  As $\Psi$ lands in a sub-theory of $\T$, $\Phi\wedge\Psi$ belongs to $\T$ and
  therefore is checkable in polynomial time with regard to the size of $\Phi$.
\end{proof}

\subsection{$\R$-absorbing functions (\Cref{prop:r-absorbing-function})}\label{proof:r-absorbing}

\emph{%
Let $f\left(t_1,\dots,t_n\right)$ be a $\R$-absorbing function of support
$\I_\R$, and let $\taint{t_i}$ be a $\sic_{t_i,\set{x}}$ for some $\set{x}$.
Then $\R\left(t_{i\in\mathcal{\I_\R}}\right)
\bigwedge_{i\in\mathcal{\I_\R}}\taint{t_i}$ is a $\sic_{f,\set{x}}$.
}

\begin{proof}
  Let $\set{x}$ be a set of variables, and let $\R$ be a relation of absorption
  with support $\I_\R$ for $f$.
  For every $i\in\I_\R$ let $\taint{t_i}$ be a $\sic_{t_i,\set{x}}$.
  By definition $\taint{t_i}\left(\set{a}\right) \models
  \forall\set{x}.\forall\set{y}.
  t\left(\set{x},\set{a}\right) = t\left(\set{y},\set{a}\right)$, and therefore
  $\bigwedge_{i\in\mathcal{\I_\R}}\taint{t_i}\left(\set{a}\right) \models
  \bigwedge_{i\in\mathcal{\I_\R}} \forall\set{x}.\forall\set{y}.
  t_i\left(\set{x},\set{a}\right) = t_i\left(\set{y},\set{a}\right)$.
  Finally, as $\R$ is a relation of absorption for $f$,
  ${\R\left(t_i\left(\set{a}\right)\right)
  {\bigwedge_{i\in\mathcal{\I_\R}}\taint{t_i}\left(\set{a}\right)}} \models
  \forall\set{x}.\forall\set{y}.
  {f\left(t_i\left(\set{x},\set{a}\right)\right) =
  f\left(t_i\left(\set{y},\set{a}\right)\right)}$, i.e.,\
  $\R\left(t_{i\in\mathcal{\I_\R}}\right)
  \bigwedge_{i\in\mathcal{\I_\R}}\taint{t_i}$ is a $\sic_{f,\set{x}}$.
\end{proof}

\section{Benchmarks} \label{appendix:sec:benchs}

\begin{table}
  \caption{Number of successes (\sat\ or \unsat) and failures (\unknown), and
  average and total resolution time in seconds}
  \label{appendix:table:all-solvers-smt-lib}
  \label{appendix:table:resolution_time}
  \makebox[\textwidth][c]{
    \begin{tabular}{|c|c|c|c|c|c|c|c|c|c|c|c|c|c|c|}\cline{5-13}
      \multicolumn{4}{c|}{\strut} &
      Boolector\raisebox{1pt}{$\bullet$} &
      CVC4 & CVC4\raisebox{1pt}{$\bullet$} &
      CVC4$_{E}$ & CVC4$_{E}$\raisebox{1pt}{$\bullet$} &
      Z3 & Z3\raisebox{1pt}{$\bullet$} &
      Z3$_{E}$ & Z3$_{E}$\raisebox{1pt}{$\bullet$} \\\hline

      \multirow{10}{*}{\rotatebox{90}{\smtlib\ formulas\strut}}
      & \multicolumn{2}{c|}{}   & \sat\strut     & \bf 399 & 84   & 242  & 84   & 242  & 261  & 366  & 87   & 366  \\
      & \multicolumn{2}{c|}{\#} & \unsat\strut   & N/A     & 0    & N/A  & 0    & N/A  & 165  & N/A  & 0    & N/A  \\
      & \multicolumn{2}{c|}{}   & \unknown\strut & 870     & 1185 & 1027 & 1185 & 1027 & 843  & 903  & 1182 & 903  \\\cline{2-13}
      & \multirow{7}{*}{\rotatebox{90}{resolution time\strut}}
      &   \multirow{3}{*}{\rotatebox{90}{average\strut}}
      &     \sat\strut     &   0.67 &   0.00 & 163.84 &   0.00 & 161.58 &  13.22 &   9.25 &   0.00 &  24.18 \\
      & & & \unsat\strut   &    N/A &    N/A &    N/A &    N/A &    N/A &   0.02 &    N/A &    N/A &    N/A \\
      & & & \unknown\strut &   0.09 &   0.14 &   1.23 &   0.14 &   1.16 &   9.20 &   0.11 &   0.16 &   1.10 \\\cline{3-13}
      & & \multirow{3}{*}{\rotatebox{90}{total\strut}}
      &     \sat\strut     & 266.98 &   0.03 &  39\,650 &   0.03 &  38\,618 & 3\,450.0 & 3\,384.7 &   0.01 & 8\,849.3 \\
      & & & \unsat\strut   &    N/A &    N/A &      N/A &    N/A &      N/A &     3.73 &      N/A &    N/A &      N/A \\
      & & & \unknown\strut &  82.39 & 165.51 & 1\,069.4 & 165.46 & 1\,009.9 & 2\,981.2 &    96.94 & 192.64 &    88.60 \\\cline{3-13}
      & & \multicolumn{2}{c|}{all (\small with \timeout)\strut}
               & 349.37 & 165.54 & 194\,667 & 165.49 & 196\,934 & 270\,150 & 36\,480 & 192.74 &  41\,935 \\\hline\hline

      \multirow{10}{*}{\rotatebox{90}{\binsec\ formulas\strut}}
      & \multicolumn{2}{c|}{}   & \sat\strut     & \bf 1042& 951  & 954  & 951  & 954  & 953  & \bf 1042& 953  & \bf 1042\\
      & \multicolumn{2}{c|}{\#} & \unsat\strut   & N/A     & 62   & N/A  & 62   & N/A  & 319  & N/A     & 62   & N/A     \\
      & \multicolumn{2}{c|}{}   & \unknown\strut & 379     & 408  & 467  & 408  & 467  & 149  & 379     & 406  & 379     \\\cline{2-13}
      & \multirow{7}{*}{\rotatebox{90}{resolution time\strut}}
      &   \multirow{3}{*}{\rotatebox{90}{average\strut}}
      &     \sat\strut     &     1.09 &     0.34 &     3.42 &   0.36 &   3.35 &     0.06 &    10.95 &   0.07 &  11.13 \\
      & & & \unsat\strut   &      N/A &     0.01 &      N/A &   0.01 &    N/A &     0.04 &      N/A &   0.01 &    N/A \\
      & & & \unknown\strut &     0.04 &    15.59 &    31.01 &  15.59 &  31.67 &   191.61 &     0.02 &   0.17 &   0.02 \\\cline{3-13}
      & & \multirow{3}{*}{\rotatebox{90}{total\strut}}
      &     \sat\strut     & 1\,138.4 &   326.98 & 3\,266.3 &   341.75 & 3\,197.7 &    55.81 &  11\,408 &  62.41 & 11\,597 \\
      & & & \unsat\strut   &    N/A   &     0.61 &      N/A &     0.39 &      N/A &    14.17 &      N/A &   0.23 &     N/A \\
      & & & \unknown\strut &  13.78   & 5\,442.4 &  12\,591 & 5\,441.9 &  12\,875 &  28\,167 &     6.15 &  73.03 &    7.05 \\\cline{3-13}
      & & \multicolumn{2}{c|}{all (\small with \timeout)\strut}
               & 1152.17 & 64\,761 & 76\,811 & 64\,772 & 77\,009 & 30\,235 &  11\,415 & 135.67  & 11\,604 \\\hline

      \multicolumn{13}{c}{
        \begin{tabular}{c@{\qquad}c}
          solver\raisebox{1pt}{$\bullet$}: solver enhanced with our method &
          Z3$_{E}$: essentially E-matching\\
        \end{tabular}} \\
    \end{tabular}}
\end{table}

\todo{@BF : trop de N/A dans table, non ?}

\end{document}